\documentclass[letterpaper, 10 pt, conference]{ieeeconf}  

\IEEEoverridecommandlockouts                              

\usepackage{amsmath} 
\usepackage{amssymb}  
\usepackage[short, c3, nocomma]{optidef}

\usepackage{graphicx}
\usepackage{amsmath}
\usepackage{amssymb}
\usepackage{amsthm}
\usepackage{dsfont}
\usepackage{commath}
\usepackage[noadjust]{cite}
\usepackage{graphicx}
\usepackage{scalerel}

\usepackage[ruled,vlined]{algorithm2e}

\usepackage[american]{circuitikz}
\usepackage{tikz}

\newtheorem{lem}{Lemma}

\theoremstyle{definition}
\newtheorem{defn}{Definition}

\newtheorem{rem}{Remark}
\newtheorem*{remark}{Remark}

\theoremstyle{definition}

\title{\LARGE \bf
Closed-Form Minkowski Sum Approximations for Efficient Optimization-Based Collision Avoidance
}

\author{James Guthrie, Marin Kobilarov,\thanks{J. Guthrie and E. Mallada are with the Department of Electrical and Computer Engineering, Johns Hopkins University, Baltimore, MD 21218, USA. Email: {\tt\small \{jguthri6, mallada\}@jhu.edu}.} Enrique Mallada \thanks{M. Kobilarov is with the Department of Mechanical Engineering, Johns Hopkins University, Baltimore, MD 21218, USA. Email: {\tt\small marin@jhu.edu}
}}

\begin{document}

\maketitle
\thispagestyle{empty}
\pagestyle{empty}

\begin{abstract}
Motion planning methods for autonomous systems based on nonlinear programming offer great flexibility in incorporating various dynamics, objectives, and constraints. One limitation of such tools is the difficulty of efficiently representing obstacle avoidance conditions for non-trivial shapes. For example, it is possible to define collision avoidance constraints suitable for nonlinear programming solvers in the canonical setting of a circular robot navigating around $M$ convex polytopes over $N$ time steps. However, it requires introducing $(2+L)MN$ additional constraints and $LMN$ additional variables, with $L$ being the number of halfplanes per polytope, leading to larger nonlinear programs with slower and less reliable solving time. In this paper, we overcome this issue by building closed-form representations of the collision avoidance conditions by outer-approximating the Minkowski sum conditions for collision. Our solution requires only $MN$ constraints (and no additional variables), leading to a smaller nonlinear program. On motion planning problems for an autonomous car and quadcopter in cluttered environments, we achieve speedups of 4.8x and 8.7x respectively with significantly less variance in solve times and negligible impact on performance arising from the use of outer approximations.
\end{abstract}


\section{INTRODUCTION}
Motion planning is a central task of most autonomous systems, including robots, drones, and autonomous vehicles.
Of the many approaches to motion planning, techniques based on nonlinear programming (NLP) such as direct multiple shooting \cite{Bock1984} and direct collocation \cite{Biegler1984} generally offer the most flexibility in regards to choice of objectives and constraints imposed.
As high-quality NLP solvers and supporting automatic differentiation tools have become available, it has become feasible to utilize these optimization-based approaches for real-time motion planning or trajectory generation \cite{Liniger2015}.


Despite the flexibility that NLP solvers provide, it can be difficult to efficiently represent obstacle avoidance constraints. Due to their reliance on gradient and Hessian information, most NLP solvers require the objective and constraints to be twice continuously differentiable expressions. This presents a challenge for collision avoidance constraints which often cannot be represented in smooth closed-forms. We briefly review two viable approaches and discuss their advantages and limitations.
\begin{figure}[]
    \centering
    \includegraphics[keepaspectratio,,width=0.485\textwidth,trim={0 0.05cm 0 0},clip]{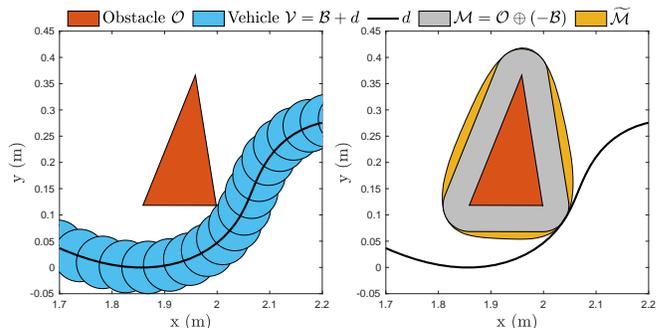}
    \vspace{-0.75cm}
    \caption{Obstacle avoidance in workspace (left) and C-space (right)}
    \vspace{-5mm}
    \label{fig:mink_traj_2x1}
\end{figure}

\noindent
\textbf{Distance Formulation:} Collision avoidance can be viewed as ensuring the minimum distance between an obstacle $\mathcal{O}$ and a vehicle $\mathcal{V}$ is greater than zero. In robotics this would be classified as performing collision checking directly in the workspace \cite{LaValle2006} as shown in the left subplot of Figure \ref{fig:mink_traj_2x1}. When the obstacle and the vehicle have convex shapes, the distance between these sets can be computed through convex optimization \cite{Boyd2004}. Using this formulation as a constraint leads to a bi-level NLP for which we lack reliable solvers. By leveraging strong duality~\cite{Zhang2021}, it is possible to reformulate the collision avoidance conditions into expressions amenable to a NLP solver. This is done at the expense of introducing new variables and constraints. In practice, the solver performance can be highly sensitive to the initialization of these variables~\cite{Zhang2018} and the resulting increase in problem complexity can be problematic for real-time motion planning in cluttered environments. Similar remarks hold for collision avoidance reformulations based on Farkas' Lemma~\cite{Gerdts2012} or polar set representations~\cite{Patel2011}.



\noindent
\textbf{Minkowski Sum Formulation:} Collision avoidance can alternatively be viewed through the lens of computational geometry as shown in the right subplot of Figure \ref{fig:mink_traj_2x1}.  Given the vehicle position $d \in \mathbb{R}^n$, and shapes $\mathcal{B},\mathcal{O} \subset \mathbb{R}^n$ of vehicle and obstacle, respectively, collision avoidance can be posed as ensuring $d \not\in \mathcal{M}$ where $\mathcal{M} = \mathcal{O} \oplus (-\mathcal{B})$, with $\oplus$ being the Minkowski sum operation~\cite{LaValle2006}.
In robotics this is often referred to as the configuration-space (C-space) approach.
Incorporating this as a constraint in an NLP solver would require a closed-form, smooth representation of the indicator function of this set. In general, this does not exist as the sets are semialgebraic, involving multiple polynomial (in)equalities. 
A notable exception is the case of bodies whose boundary surface are smooth and admit both implicit and parametric representations~\cite{Ruan2022}, which includes ellipsoids and convex superquadrics \cite{Ruan2019}.
However, many implicit surfaces do not admit a parametric representation and for others obtaining one is an open problem \cite{Hoffmann1990}. Additionally, this approach cannot address the practical case of non-smooth boundaries such as convex polytopes.

\newcommand{\Mouter}{\widetilde{\mathcal{M}}}

\subsection{Contributions}
In this work we propose efficient collision avoidance conditions based on closed-form, outer approximations $\Mouter \supseteq \mathcal{M}$ of the Minkowski sum. We focus on the important case in which the obstacle is a bounded, convex polytope and the vehicle is represented by Euclidean balls (possibly multiple). Building upon recent successes of sum-of-squares (SOS) optimization for outer approximating semialgebraic sets \cite{Magnani2005, Dabbene2017, Ahmadi2017, Guthrie2022Star, Guthrie2021}, we develop SOS programs for finding $\Mouter$. Figure \ref{fig:mink_traj_2x1} shows an example of the resulting outer approximations. We then use $\Mouter$ to perform optimization-based motion planning of an autonomous car and quadcopter navigating cluttered environments. Compared to the exact method~\cite{Zhang2021}, our approximate method solves 4.8x (car) and 8.7x (quadcopter) faster while introducing minimal conservatism arising from the use of outer approximations.


The rest of the paper is organized as follows. Section II reviews relevant aspects of convex sets, Minkowski sums, and SOS optimization. Section III defines the motion planning problem. Section IV poses the obstacle avoidance constraints using Minkowski sums and provides methods for outer approximating the set. Section V applies our approach to motion planning for an autonomous car and quadcopter. Section VI concludes the paper with a discussion of future directions.

\newcommand{\Set}[1]{\mathcal{R}_{#1}}
\newcommand{\Point}[1]{r_{#1}e^{i\theta_{#1}}}
\newcommand{\CPoint}[1]{\mathbf{r}_{#1}}
\newcommand{\Zero}{{0}}
\newcommand{\hull}{\bold{conv}}
\newcommand{\sos}[1]{\textstyle\sum[#1]}
\newcommand{\obs}{\mathcal{O}}

\newcommand{\TrajOpt}{(14)}
\newcommand{\TrajOptDyn}{(14c)}

\section{PRELIMINARIES}
We briefly review some basic properties of convex sets, Minkowski sums, and sum-of-squares polynomials. This is mostly done to setup our notation. The reader is referred to \cite{Boyd2004, Schneider2013, Parillo2000} for proofs and further details.
\subsection{Set Definitions}

\begin{defn}[Convex Hull]\label{defn:cvxhull}
The convex hull of a set $\mathcal{B}$ is defined as: 
$\bold{conv}\;\mathcal{B} = \{\theta_1x_1 + \hdots +\theta_kx_k\ \, | \,  x_i \in \mathcal{B}, \theta_i \geq 0, i = 1,\hdots,k,\sum_{i=1}^k\theta_i = 1\}$. Let $\mathcal{C}$ be any convex set that contains $\mathcal{B}$. The convex hull is the smallest convex set that contains $\mathcal{B}$:
\begin{equation}\label{eqn:cvxhullproperty}
    \mathcal{B} \subseteq \mathcal{C} \Leftrightarrow \bold{conv}\;\mathcal{B} \subseteq \mathcal{C}
\end{equation}
\end{defn}

\begin{defn}[$\alpha$-sublevel Set]\label{defn:alphasublevel}
The $\alpha$-sublevel set of a function $f:\mathbb{R}^n \rightarrow \mathbb{R}$ is: $B_\alpha = \{x \,|\, f(x) \leq \alpha\}$. We denote the boundary of the set as $\partial B_\alpha = \{x \,|\, f(x) = \alpha\}$.
\end{defn}

\begin{lem}\label{lem:dBhull}
Let $\mathcal{B}_\alpha$ be a convex set that is the $\alpha$-sublevel set of a function $f:\mathbb{R}^n \rightarrow \mathbb{R}$. Then $\mathcal{B}_\alpha = \bold{conv}\; \partial \mathcal{B}_\alpha$.
\end{lem}

We use the notation $-\mathcal{B} = \{-b \,|\, b \in \mathcal{B} \}$ to represent the set $\mathcal{B}$ reflected about the origin. Note that $-\mathcal{B}$ is convex if and only $\mathcal{B}$ is convex.

\begin{defn}[Polytope]\label{defn:polytope}
A polytope $\mathcal{P}$ is defined as the solution set of $j$ linear inequalities in $\mathbb{R}^n$. This set is convex by construction. We impose the additional requirement that the set is bounded. The linear inequalities give the halfspace representation
\begin{equation}\label{eqn:polytopedef_halfspace}
    \mathcal{P} = \{x \,|\, Ax \leq b \}
\end{equation}
where $A \in \mathbb{R}^{j \times n}, b \in \mathbb{R}^j$. Alternatively, the polytope can be represented by the convex hull of its $k$ vertices
\begin{equation}\label{eqn:polytopedef_vertex}
    \mathcal{P} = \bold{conv}\;\{v_1, v_2,\hdots,v_k\}
\end{equation}
where $v_i \in \mathbb{R}^n, i \in [k]:= \{1,\hdots,k\}$.
\end{defn}

\subsection{Minkowski Sum Properties}
\begin{defn}[Minkowski Sum] Given two sets $\mathcal{A}, \mathcal{B}$, their Minkowski sum is defined as follows:
\begin{equation}
    \mathcal{A}\oplus\mathcal{B} = \{a + b \,|\, a \in\mathcal{A}, b\in\mathcal{B}\}
\end{equation}
\end{defn}\label{defn:minksum}

\begin{lem}\label{lem:minksum}
If $A$ and $B$ are convex sets then $A \oplus B$ is convex.
\end{lem}

\begin{lem}\label{lem:minkdistributive}
For any sets $\mathcal{A}, \mathcal{B}$ the following equality holds:
\begin{equation}
    \bold{conv}\;(A \oplus B) = \bold{conv}\;(A) \oplus \bold{conv}\;(B)
\end{equation}

\end{lem}

\subsection{Sum-of-Squares Optimization}
For $x \in \mathbb{R}^n$, let $\mathbb{R}[x]$ denote the set of polynomials in $x$ with real coefficients. 
\begin{defn}[Sum-of-Squares Polynomial]\label{defn:sospoly}
A polynomial $p(x) \in \mathbb{R}[x]$ is a sum-of-squares (SOS) polynomial if there exists polynomials $q_i(x) \in \mathbb{R}[x], i \in [j]$ such that $p(x) = \underset{i \in [j]}{\sum} q_i^2(x)$. We use $\sos{x}$ to denote the set of SOS polynomials in $x$.  A polynomial of degree $2d$ is a SOS polynomial if and only if there exists a positive semi-definite matrix $P$ (the Gram matrix) such that $p(x) = z(x)^TPz(x)$ where $z(x)$ is the vector of all monomials of $x$ up to degree $d$ \cite{Parillo2000}. 
\end{defn}
Note that a polynomial being SOS is a sufficient condition for the polynomial to be non-negative (i.e. $p(x) \geq 0 \, \forall \, x$).
\begin{defn}[SOS-Convex]\label{defn:soscvx}
A polynomial $p(x)$ is SOS-convex if the following holds
\begin{equation}
    u^T\nabla^2p(x)u \in \sos{x,u}
\end{equation}
where $u, x \in \mathbb{R}^n$. SOS-convexity is a sufficient condition for the Hessian of $p(x)$ to be positive semi-definite and therefore $p(x)$ to be convex.
\end{defn}

 In the development that follows, we will be interested in solving slight variations of the following problem.
 
 \begin{mini!}|s|
{P}{-{\log\det}P} {}{}
\addConstraint{P \succeq 0, \quad}{p(x) = z(x)^TPz(x), \,}{\label{con:sos}}
\addConstraint{1 - p(x)}{\geq 0 \quad \forall \, x \in \mathcal{X},}{\label{con:set_contain}}
\end{mini!}
Here $\mathcal{X}$ is a semialgebraic set defined by $n_i$ polynomial inequalities and $n_j$ polynomial equalities.
 \begin{equation}
     \mathcal{X} = \{x \,|\, g_{i}(x) \geq 0, i \in [n_i], h_{j}(x) = 0, j \in [n_j] \}
 \end{equation}
 Equation \eqref{con:sos} constrains $p(x)$ to be a SOS polynomial. Equation \eqref{con:set_contain} is a set-containment condition. The generalized $\mathcal{S}$-procedure provides a sufficient condition for the set-containment to hold \cite{Parillo2000}. For each polynomial equality $g_i(x)$ or inequality $h_j(x)$ describing the set $\mathcal{X}$, we introduce a non-negative polynomial $\lambda_i(x)$ or polynomial $\mu_j(x)$ respectively. The generalized $\mathcal{S}$-procedure involves replacing \eqref{con:set_contain} with the following:
 \begin{align}
     1 - p(x) - \displaystyle\sum_i\lambda_i(x)g_i(x) - \displaystyle\sum_j\mu_j(x)h_j(x) \geq 0, \label{con:poly1}\\
     \lambda_i(x) \geq 0 \quad i \in [n_i] \label{con:poly2}
 \end{align}
By replacing the non-negativity constraints in \eqref{con:poly1}, \eqref{con:poly2} with the more restrictive condition that the expressions be SOS polynomials, we obtain a semidefinite program which is readily solved.  
 \begin{mini}|s| 
{P, \lambda_{[1:n_i]}(x), \mu_{[1:n_j]}(x)}{-{\log\det}P} {}{}
\addConstraint{P \succeq 0, \quad}{p(x) = z(x)^TPz(x), \,}{}
\addConstraint{1 - p(x) - \displaystyle\sum_i\lambda_i(x)g_i(x) - \displaystyle\sum_j\mu_j(x)h_j(x)}{\in \displaystyle\sum[x],}{\label{con:set_contain_sproc}}
\addConstraint{\lambda_i(x)}{\in \sos{x}, \quad i\in [n_i]\,.} {}
\end{mini}
Note when a polynomial is listed as a decision variable, e.g., $\lambda_{[1:n_i]}(x)$ and $\mu_{[1:n_j]}(x)$  underneath the $\min$, it is implied that the monomial basis is specified and the coefficients are decision variables.
\begin{rem}\label{rem:sosequality}
Representing an equality constraint requires introducing a polynomial $\mu(x)$. In contrast, representing an inequality requires introducing a SOS polynomial $\lambda(x)$ which has a smaller feasible set and creates an additional semidefinite constraint. As such, it is generally advantageous to represent sets using equalities when applying the generalized $\mathcal{S}$-procedure.
\end{rem}
In the development that follows we focus on transforming problems of interest into the form of (7). Once in this form, the subsequent application of the generalized $\mathcal{S}$-procedure is mechanical.

\section{Problem Statement}
We now setup the problem of optimization-based motion planning with collision avoidance constraints. For convenience, our notation closely follows that of \cite{Zhang2018}.

\subsection{Vehicle and Obstacle Models}
Consider a vehicle with states $x_k \in \mathbb{R}^{n_x}$ and inputs $u_k \in \mathbb{R}^{n_u}$ at time step $k$. The dynamics evolve according to $x_{k+1} = f(x_k, u_k)$ where $f: \mathbb{R}^{n_x} \times \mathbb{R}^{n_u} \rightarrow \mathbb{R}^{n_x}$. The vehicle occupies space in $\mathbb{R}^n$. The vehicle's shape is assumed to be represented by $n_b$ Euclidean balls with radii $r^{(i)}$.
\begin{equation}
    \mathcal{B}^{(i)} = \{ y \in \mathbb{R}^n \,|\, \|y\|_2 \leq r^{(i)} \},\; i \in [n_b].
\end{equation}
The center of each ball is a function of the vehicle's state as given by $t^{(i)}: \mathbb{R}^{n_x} \rightarrow \mathbb{R}^n$. At time index $k$, the space occupied by ball $i$ is given by:
\begin{equation}
    \mathcal{V}^{(i)}(x_k) = \mathcal{B}^{(i)} + t^{(i)}(x_k)\,.
\end{equation}
The union $\underset{i}{\bigcup}\mathcal{V}^{(i)}(x_k)$  gives the total space occupied by the vehicle at time index $k$. {For  ease of exposition, in what follows} we focus w.l.o.g. on the case when the vehicle is represented by a single ball and drop the superscript $(i)$.

We assume there are $M$ obstacles present in the environment indexed by $m\in[M]$. Each obstacle $\mathcal{O}^{(m)}$ is a polytope (closed, convex) with $k^{(m)}$ vertices $\{v_1,\hdots,v_{k^{(m)}}\}$ defining the convex hull as in \eqref{eqn:polytopedef_vertex}. Equivalently represented in halfspace form \eqref{eqn:polytopedef_halfspace}, the obstacle $m$ is defined by $j^{(m)}$ constraints given by $A^{(m)} \in \mathbb{R}^{j^{(m)} \times n}, b^{(m)} \in \mathbb{R}^{j^{(m)}}$.

\subsection{Optimal Control Problem}
We consider an optimal control problem of controlling the vehicle over $N$ steps. The vehicle starts at state $x_S$ and must end at final state $x_F$. Let $X, U$ denote the vector of all states and controls respectively, $X = [x_0^T,\hdots, x_{N}^T]^T, U = [u_0^T,\hdots,u_{N-1}]^T$. We seek to minimize an objective $l(X,U)$ where $l: X \times U \rightarrow \mathbb{R}$. Additionally, the vehicle is subject to $n_h$ constraints given by $h(X,U) \leq 0$ where $h: X \times U \rightarrow \mathbb{R}^{n_h}$ and the inequality is interpreted element-wise. We assume that $l(X,U)$ and $h(X,U)$ are continuously differentiable and therefore suitable for nonlinear programming solvers which utilize gradient and Hessian information. Lastly, we enforce collision avoidance constraints between each obstacle and the vehicle. The resulting optimization problem takes the following form:

\begin{mini!}|s|
{X,U}{l(X,U)} {}{}
\addConstraint{x_0 = x_S, \quad}{x_N = x_F,}{}
\addConstraint{x_{k+1}}{=f(x_k,u_k),}{k = 0,\hdots,N-1}
\addConstraint{h(X,U)}{\leq 0,}{}
\addConstraint{\mathcal{V}(x_k) \cap \mathcal{O}^{(m)}}{=\emptyset,}{k \in [N], \, m \in[M]. \label{con:collision_set}} 
\end{mini!}
Equation \eqref{con:collision_set} represents the collision avoidance constraints which are non-convex and non-smooth in general. In \cite{Zhang2021}, the authors provide an exact, smooth reformulation of these constraints. As the distance between two convex shapes can be computed using convex optimization, the authors leverage strong duality to develop necessary and sufficient conditions for a Euclidean ball of radius $r$ to not intersect a given convex shape. This requires introducing dual variables associated with the halfspace constraints representing each obstacle $\lambda_k^{(m)} \in \mathbb{R}^{j^{(m)}}$, $k\in[N]$, $m \in[M]$ and replacing \eqref{con:collision_set} with the following constraints.
\begin{align}\label{eqn:OBCA}
\begin{split}
    (A^{(m)}t(x_k) - b^{(m)})^T\lambda_k^{(m)} > r, \\
  \| A^{(m)^T} \lambda_k^{(m)} \| ^2_2 \leq 1, \\
   \lambda_k^{(m)} \geq 0,\\
   k \in [N], \, m \in [M].
\end{split}
\end{align}
If each obstacle has $L$ halfspace constraints, this method introduces $(2+L)MN$ constraints and $LMN$ dual variables which can result in a large nonlinear program that is computationally intensive. In the following, we present a method for approximating the collision avoidance constraints while introducing only $MN$ constraints and no additional variables.

\newcommand{\Sapprox}{\widetilde{\mathcal{M}}}
\section{Collision Avoidance via Minkowski Sums}
We will utilize Minkowski sums to represent the collision avoidance constraints between a closed, convex polytope obstacle $\mathcal{O} = \{y \in \mathbb{R}^n \,|\, a_i^Ty \leq b_i, i \in [L] \}$ and a vehicle with shape given by the Euclidean ball $\mathcal{B} = \{w \in \mathbb{R}^n\,|\, w^Tw \leq r^2\}$. We first review a fundamental result from computational geometry.  
\begin{lem}\label{lem:minkcollision}
Let $\mathcal{O}$ and $\mathcal{B}$ be sets in $\mathbb{R}^n$. Let $\mathcal{V} = \mathcal{B} + d$ be the set $\mathcal{B}$ translated by $d \in \mathbb{R}^n$. Then the following relation holds:
\begin{equation}
    \mathcal{O} \cap \mathcal{V} \neq \emptyset \Leftrightarrow d \in \mathcal{O} \oplus (-\mathcal{B})
\end{equation}
\end{lem}
\begin{proof}
See, e.g. \cite{LaValle2006, Berg2008}
\end{proof}
In words, when $\mathcal{B}$ is located at position $d$, it makes contact with $\mathcal{O}$ if and only if $d$ is in the Minkowski sum $\mathcal{O} \oplus (\mathcal{-B})$. Thus collision avoidance with respect to obstacle $\mathcal{O}$ is equivalent to ensuring $d \not\in \mathcal{O} \oplus (\mathcal{-B})$.

When $\mathcal{O}$ is a polytope and $\mathcal{B}$ is a Euclidean ball, the set $\mathcal{O} \oplus (-\mathcal{B})$ is semialgebraic. As such we cannot directly include the condition $d \not\in \mathcal{O} \oplus (-\mathcal{B})$ as a constraint in a nonlinear optimization problem which requires closed-form, twice differentiable expressions.

Instead we propose to find an outer approximation $\mathcal{O} \oplus (-\mathcal{B}) \subseteq \Sapprox \subset \mathbb{R}^n$ where $\Sapprox$ is defined as the 1-sublevel set of a function $p: \mathbb{R}^n \rightarrow \mathbb{R}$. 
Recall in our setting the translation of the ball at time index $k$ is a function of the vehicle's state $x_k$ as given by $t: \mathbb{R}^{n_x} \rightarrow \mathbb{R}^n$. Collision avoidance with respect to obstacle $\obs$ can then be ensured by imposing the constraint $p(t(x_k)) > 1 \Leftrightarrow t(x_k) \not\in \Sapprox \Rightarrow t(x_k) \not\in \mathcal{O} \oplus (-\mathcal{B}) \Leftrightarrow \mathcal{O} \cap \mathcal{V} = \emptyset$. 

If multiple obstacles $\obs^{(m)}, m \in[M]$ are present, we repeat this process for each obstacle and denote the associated function as $p^{(m)}(x)$. In our trajectory optimization problem we replace \eqref{con:collision_set} with $MN$ constraints.
\begin{equation}
    p^{(m)}(t(x_k)) > 1, \quad k \in[N], \quad m \in[M]. \label{eqn:polyobs}
\end{equation}

\subsection{Outer Approximations of the Minkowski Sum}
We would like our outer approximations to closely approximate the true set. To do so, we pose an optimization problem in which we minimize the volume of the outer approximation. 

\begin{mini}|s|
{p(x)}{\text{vol }\Sapprox} {}{}
\addConstraint{1 - p(x)}{\geq 0}{\forall \, x \in \mathcal{O} \oplus (-\mathcal{B}),}
\addConstraint{\Sapprox}{= \{x \,|\, p(x) \leq 1\}}{}
\label{opti:logdet_semialg} 
\end{mini}


In general we cannot solve this optimization problem. To arrive at a tractable formulation, we apply the generalized $\mathcal{S}$-procedure. We first parameterize the polynomial as $p(x) = z(x)^TPz(x)$ where $z(x)$ is a monomial basis chosen by the user and $P$ is a positive semi-definite matrix of appropriate dimension.  For arbitrary polynomials, we lack an expression for minimizing the volume of the 1-level set.  Various heuristics have been proposed \cite{Magnani2005, Dabbene2017, Ahmadi2017, Guthrie2022Star}. We have found maximizing the determinant of $P$, as proposed in \cite{Ahmadi2017}, to work well for the problems herein. The resulting optimization problem is
\begin{mini}|s|
{P}{-{\log\det}P} {}{}
\addConstraint{p(x) }{= z(x)^TPz(x),}{\quad P\succeq 0,}
\addConstraint{1 - p(x)}{\geq 0}{\,\forall \, x \in \{y - w \,|\, a_i^Ty \leq b_i,}
\addConstraint{}{}{\quad \quad \quad w^Tw \leq r^2,i \in[L]\}}
\label{opti:outerapprox} 
\end{mini}
where we have explicitly written the set resulting from the Minkowski sum in terms of $y$ and $w$ along with inequalities that ensure $y \in \obs$ and $w \in \mathcal{B}$. 

We apply the $\mathcal{S}$-procedure to replace the set-containment condition with a sufficient condition. This requires introducing multipliers $\lambda(y,w)$. We then replace the non-negativity conditions with the sufficient condition that the expression admits a SOS decomposition in terms of variables $y$ and $w$.

\noindent\textbf{Optimization Problem 1: Outer Approximation}
\begin{mini*}|s|
{P, \lambda_{[0:L]}(y,w)}{-{\log\det}P} {}{}
\addConstraint{p(x) = z(x)^TPz(x),}{\quad P\succeq 0,}{\tag{OA} \label{opti:outerapprox}}
\addConstraint{1 - p(y-w)- \lambda_0(y,w)(r^2 - w^Tw)}{}{}
\addConstraint{\quad \quad -\sum_{i=1}^L\lambda_i(y,w)(b_i - a_i^Ty)}{\in \sum[y,w]}{}
\addConstraint{\lambda_i(y,w)}{\in \sos{y,w} \quad i = 0,\hdots,L}{}
\end{mini*}
The formulation given by \eqref{opti:outerapprox} is viable but computationally expensive because the SOS decompositions involve both $w$ and $y$ giving $2n$ free variables for $x \in \mathbb{R}^n$. As we seek higher-order approximations, the monomial basis grows rapidly in size leading to large semidefinite programs. We now develop a computationally cheaper program by leveraging convexity.

\subsection{Convex Outer Approximations of the Minkowski Sum}
In developing an efficient method for outer approximating the Minkowski sum, we will utilize the following Lemma.
\begin{lem}\label{lem:minksum_extension}
Let $\mathcal{O} \subset \mathbb{R}^n$ be a polytope with $K$ vertices $\{v_i\}, i \in [K]$. Let $\mathcal{B} \subset \mathbb{R}^n$ be a convex set that is the $\alpha$-sublevel set of a function $f:\mathbb{R}^n \rightarrow \mathbb{R}$. Let $S$ be any convex set in $\mathbb{R}^n$. Then the following relation holds:
\begin{equation}
    \mathcal{O} \oplus (-\mathcal{B}) \subseteq S \Leftrightarrow \{v_i\} \oplus (-\partial \mathcal{B}) \subseteq S
\end{equation}
\begin{proof}
First represent $\mathcal{O} \oplus (-\mathcal{B})$ in terms of its convex hull.
\begin{align*}
   \mathcal{O} \oplus (-\mathcal{B}) &= \bold{conv}\{v_i\} \oplus \hull(-\partial \mathcal{B}), &&\textbf{Lemma \ref{lem:dBhull}}   \\
&= \hull[\{v_i\} \oplus (-\partial \mathcal{B})], &&\textbf{Lemma \ref{lem:minkdistributive}}
\end{align*}
Next apply the property of the convex hull \eqref{eqn:cvxhullproperty}.
\begin{align*}
   \hull[\{v_i\} \oplus (-\partial \mathcal{B})] \subseteq S \Leftrightarrow \{v_i\} \oplus (-\partial \mathcal{B}) \subseteq S
\end{align*}
\end{proof}

\end{lem}

Lemma \ref{lem:minksum_extension} provides a more efficient condition for outer approximating the Minkowski sum with a set $\Sapprox = \{x \,|\, p(x) \leq 1\}$ as we only have to consider the vertices of $\obs$ and the boundary of $\mathcal{B}$ in our set-containment constraint as follows.
\begin{equation}
1 - p(x) \geq 0 \quad \forall \, x \in {v_i} \oplus (-\partial \mathcal{B}) \quad i \in[K].
\end{equation}

However, it requires the condition that $\Sapprox$ be a convex set. We argue that this is a reasonable constraint as $\mathcal{O} \oplus (-\mathcal{B})$ is itself convex. It is difficult to impose the condition that the 1-sublevel set of $p(x)$, i.e., $\Sapprox$, is convex. Recalling Definition \ref{defn:soscvx}, we will instead impose the sufficient condition that the function $p(x)$ be sos-convex. 

As done previously, we rewrite the set-containment conditions using the generalized $\mathcal{S}$-procedure. We then replace the non-negativity conditions with SOS conditions.

\noindent\textbf{Optimization Problem 2: Convex Outer Approximation}
\begin{mini*}|s|
{P, \mu_{[1:K]}(w)}{-{\log\det}P} {}{}
\addConstraint{P \succeq 0, \,}{p(x) = z(x)^TPz(x)}{\tag{COA}\label{opti:outerapproxcvx}}
\addConstraint{1 - p(v_i - w) - \mu_i(w)(r^2 - w^Tw)}{\in \sos{w},}{\, i \in [K]}
\addConstraint{u^T\nabla^2p(x)u}{\in \sos{x,u}}{}
\end{mini*}

\begin{rem}
The formulation of \eqref{opti:outerapproxcvx} is advantageous in that the multipliers $\mu(w)$ do not have to be SOS and they only depend on $n$ free variables ($w \in \mathbb{R}^n$). In contrast, \eqref{opti:outerapprox} requires SOS multipliers $\lambda(w,y)$ which depend on $2n$ free variables ($w, y \in \mathbb{R}^n$). The former leads to smaller semidefinite programs which scale better with respect to the dimension $n$ or the complexity of $\mathcal{O}$. This is numerically illustrated in the following example.
\end{rem}

\subsection{2D Example}
We generate 1000 random test cases in $\mathbb{R}^2$. For each case we generate a polytope $\mathcal{O}$ with $n \in \{3,4,\hdots,12\}$ vertices $v_i \in [-1,1]^2, i \in [n]$ along with a disk $\mathcal{B}$ with radius $r \in [0,1]$. We form outer approximations $\Sapprox = \{x \,|\, p(x) \leq 1\}$ of the set $\mathcal{M} = \mathcal{O} \oplus (-\mathcal{B})$ using both \eqref{opti:outerapprox} and \eqref{opti:outerapproxcvx}. For each we consider polynomials $p(x)$ of degree 2, 4 and 6.  
To assess the accuracy of our outer approximations, we compute the approximation error as $100 \times \frac{\textup{Area}(\Sapprox)-\textup{Area}(\mathcal{M})}{\textup{Area}(\mathcal{M})}$. Table \ref{tab:1} lists the mean approximation error of the 1000 test cases. Empirically, as we increase the polynomial order, the approximation error is reduced, indicating we are getting better outer approximations. Table \ref{tab:2} lists the mean solve times. As expected, \eqref{opti:outerapproxcvx} has significantly faster solve times than \eqref{opti:outerapprox} due to the smaller semidefinite program. Figure \ref{fig:fit_example_2D_twocol} provides an example of the results.
\begin{figure}[]
    \centering
    \vspace{0.2cm}
    \includegraphics[width=0.45\textwidth]{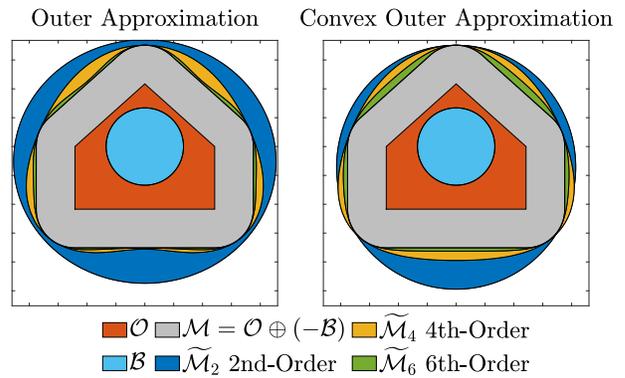}
    \vspace{-0.25cm}
    \caption{Outer Approximation of Minkowski Sum}
    \label{fig:fit_example_2D_twocol}
\end{figure}

\begin{table}[ht] 
\caption{Mean Approximation Error of Minkowski Sums}
\centering
\begin{tabular}{||c | c c c||}
 \hline
 Polynomial Degree & 2 & 4 & 6 \\ [0.5ex] 
 \hline
 Outer Approximation & 40\% & 9\% & 3\% \\ 
 \hline
 Convex Outer Approximation & 25\% & 9\% & 5\% \\
 \hline
\end{tabular}
\label{tab:1}
\end{table}

\begin{table}[!ht]
\caption{Mean Solve Times (s) of Optimization Problems 1 \& 2}
\centering
\begin{tabular}{||c | c c c||}  
 \hline
 Polynomial Degree & 2 & 4 & 6 \\ [0.5ex] 
 \hline
 Outer Approximation & 0.020 & 0.174 & 0.925 \\ 
 \hline
 Convex Outer Approximation & 0.004 & 0.014 & 0.049 \\
 \hline
\end{tabular}
\label{tab:2}
\end{table}
\begin{remark}
For the case when $p(x)$ is a quadratic, the resulting minimum volume $\Sapprox$ can be found exactly using the semidefinite program for finding the minimum volume outer ellipsoid (MVOE) covering a union of ellipsoids \cite{Boyd2004}. In this case each ellipsoid is a ball of radius $r$ centered at vertex $v_i$. Our convex formulation \eqref{opti:outerapproxcvx} can be seen as a generalized form of this result. The non-convex case \eqref{opti:outerapprox} has a smaller feasible set due to the reliance on SOS multipliers and does not return the minimum volume outer ellipsoid in general. Thus \eqref{opti:outerapprox} is only advantageous when seeking non-ellipsoidal approximations.
\end{remark}

\section{Motion Planning Examples}
We demonstrate our proposed obstacle avoidance conditions on an autonomous car and quadcopter example. We solve \TrajOpt~using both the exact representation \eqref{eqn:OBCA} and the approximate representation \eqref{eqn:polyobs} of the collision avoidance constraints. We compute the sub-optimality of the approximate method relative to the exact method as $100 \times \frac{J_{approx}-J_{exact}}{J_{exact}}$ where $J_{approx}, J_{exact}$ are the value of the objective function $l(X,U)$ for the respective solutions. 

In each example, the dynamic constraints \TrajOptDyn ~are implemented using a 4th-order Runge-Kutta integrator with a time-step of $0.02s$ over $N = 150$ steps giving a 3s time horizon. We use $A^\star$ \cite{Hart1968} to find a minimum-distance collision free path on a discretized representation of the environment.\footnote{The $A^\star$ computation time takes an average of 7ms in the case of 10 obstacles for the car example and 41ms for the quadcopter example. In both cases the $A^\star$ step is less than 1/10th the NLP runtime. Bypassing this guess generation and using a naive initialization by linearly interpolating from the initial state $x_S$ to the final state $x_F$ resulted in poor solver reliability.} This path does not consider the dynamics and is generally not kinematically feasible. We use this to initialize our guess for the vehicle's states over time. The approximate representation utilizes 4th-order, convex polynomials to represent the Minkowski sums. For the exact method, similar to \cite{Zhang2021}, we initialize the dual variables $\lambda$ to 0.05. 

\subsection{Autonomous Car}
We adopt the autonomous racing car model from \cite{Liniger2015}. The model has 6 states, $x = \begin{bmatrix} p_x & p_y & \psi & v_x & v_y & \omega \end{bmatrix}^T$ consisting of position $(p_x, p_y)$, orientation $(\psi)$, body velocities $(v_x, v_y)$ and yaw rate $(\omega)$. The inputs are motor duty cycle $(d)$ and steering angle $(\delta)$. We represent the vehicle's shape as a single disk $\mathcal{B}$ of radius $r = 0.05\text{m}$. The center of the disk at time step $k$ is the $(p_{x,k},p_{y,k})$ position of the vehicle: $t(x_k) = \begin{bmatrix} p_{x,k} & p_{y,k} \end{bmatrix}^T$.

We consider a situation in which the vehicle is making forward progress along a straight track while navigating obstacles. The objective is to minimize the 2-norm of the input, $l(X,U) = \|U\|_2^2$. The vehicle starts at state $x_S = \begin{bmatrix} 0 & p_{y,S} & 0 & 1 & 0 & 0 \end{bmatrix}^T$ and must end at position $(3, p_{y,F})$. At each step $k = 0,\hdots,N-1$, the vehicle is subject to box constraints on the position $p_{x,k} \in [0, 3], p_{y,k} \in [0, 0.3]$ and inputs $d_{k} \in [-0.1, 1], \delta_k \in [-1, 1]$. 

\begin{figure*}[]
    \centering
    \vspace{0.2cm}
    \includegraphics[width=0.99\textwidth]{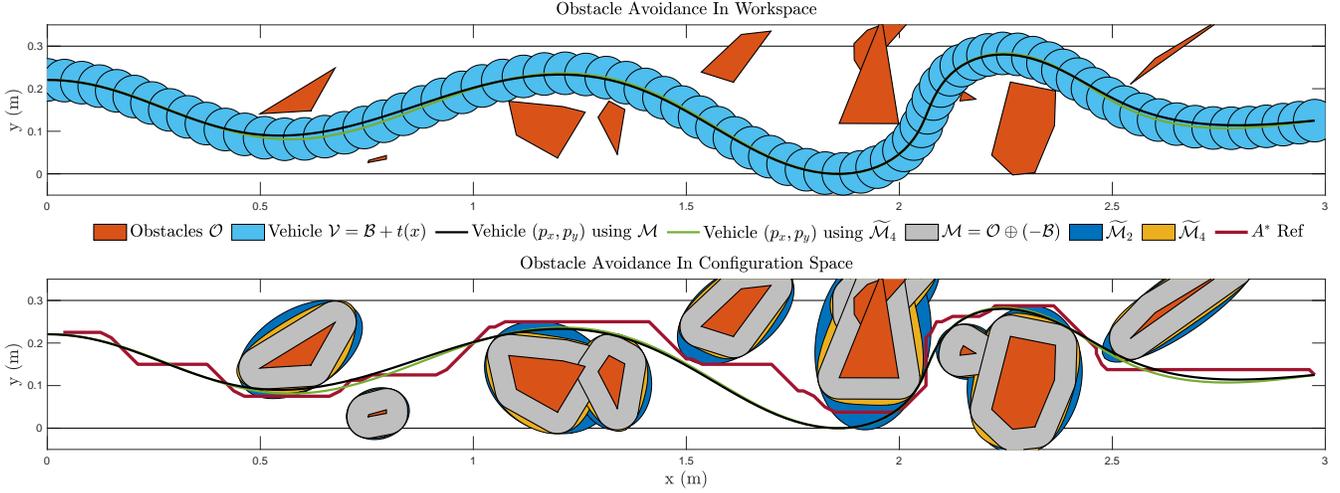}
    \vspace{-0.35cm}
    \caption{Autonomous car navigating obstacles in workspace (upper) and configuration space (lower)}
    \label{fig:car_traj_2x1}
    \vspace{-0.25cm}
\end{figure*}
We consider scenarios consisting of $M \in \{1,2,\hdots,10\}$ obstacles. For each scenario, we generate 100 random test cases in which we vary the start and final $y$ position, $p_{y,S}, p_{y,F} \in [0, 0.3]$ along with the placement and shapes of the $M$ obstacles. Figure \ref{fig:car_traj_2x1} shows a scenario in which the vehicle navigates ten obstacles. We plot the obstacles along with the exact and approximate Minkowski sums of each obstacle and the vehicle $\mathcal{B}$. As the exact method is equivalent to ensuring the vehicle's position ($p_x, p_y)$ remains outside the exact Minkowski sums $\mathcal{M}$, this helps to visualize the conservatism of our outer approximations.  The 4th-order approximate representations $\widetilde{\mathcal{M}}_4$ are quite tight and are only visible as thin yellow borders around the exact Minkowski sums in gray. For reference, we also plot 2nd-order, ellipsoidal approximations $\widetilde{\mathcal{M}}_2$ which are unacceptably conservative as the vehicle cannot progress beyond $p_x = 1.9$ without violating constraints. As the objective penalizes large steering and acceleration commands, the vehicle naturally makes tight maneuvers around the obstacles. The exact method returns a slightly better trajectory because the configuration-space obstacles it must avoid are smaller, requiring less maneuvering. The approximate method makes slightly wider turns, resulting in a 4\% sub-optimal trajectory. However the difference is minor and the resulting trajectories are nearly identical. 

\begin{figure}[]
    \centering
\includegraphics[width=0.45\textwidth]{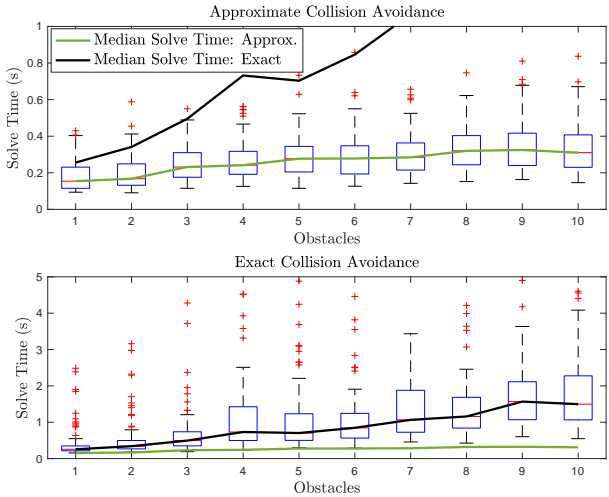}
    \vspace{-0.35cm}
    \caption{Solve time statistics for autonomous car example.}
    \label{fig:solve_times_2d}
    \vspace{-0.5cm}
\end{figure}
Figure \ref{fig:solve_times_2d} shows the solve time statistics of the approximate and exact methods as we vary the number of obstacles present. Comparing median solve times, the approximate method solves 1.6x faster than the exact method when just one obstacle is present. With ten obstacles present, the approximate method solves 4.8x faster. The approximate method shows less variability in the solve times, with a maximum solve time of 0.84s and no failed instances.\footnote{
The solver times reported for the approximate method only reflect the time spent solving the nonlinear program. We do not include the time required to compute the outer approximations. In a real-time motion planning problem, these approximations would only be performed once per obstacle, either offline or online. We note that based on Table \ref{tab:2}, approximating an obstacle with a 4th-order convex polynomial takes 0.014s. If this computation time were included in Figure \ref{fig:solve_times_2d}, the approximate method would still be consistently faster than the exact method. 
} For the exact method, 240 of the 1000 cases either did not converge or exceeded the maximum allowed solve time of 5s. For 746 of the 760 cases in which the exact method was successfully solved, the approximate method returned a trajectory less than 5\% sub-optimal. The worst-case sub-optimality was 18\%.

\subsection{Quadcopter}
We consider the quadcopter model from \cite{Mellinger2012}. The model has 12 states consisting of position $(p_x, p_y, p_z)$, velocity $(v_x, v_y, v_z)$, Euler angles $(\phi, \theta, \psi)$, and body rates $(p, q, r)$. The inputs are the four rotor speeds $\omega_i, i \in [4]$ in scaled values. We represent the quadcopter's shape as a single ball $\mathcal{B}$ of radius $r = 0.25\text{m}$. The center of the ball at time step $k$ is the position of the quadcopter: $t(x_k) = \begin{bmatrix} p_{x,k} & p_{y,k} & p_{z,k} \end{bmatrix}^T$.

We consider a situation in which the quadcopter is navigating a cluttered room with dimensions $10 \times 10 \times 5$. The quadcopter starts at the origin with state $x_S = \mathbf{0}_{12}$ and must end at state $x_F = \begin{bmatrix} p_{x,F} & p_{y,F} & p_{z,F} & \mathbf{0}_9 \end{bmatrix}^T$ while avoiding any obstacles. Here $\mathbf{0}_{i}$ denotes the zero vector in $\mathbb{R}^i$. The objective is to minimize the 2-norm of the rotor speed deviation from a trim condition $l(X,U) = \|U-4.5\|_2^2$ where $\omega_i = 4.5, i \in [4]$ achieves a steady-state, hover condition. At each step $k = 0,\hdots,N-1$, the vehicle is subject to box constraints on the position $p_{x,k} \in [0, 10], p_{y,k} \in [0, 10], p_{z,k} \in [0, 5]$ and inputs $\omega_{i,k} \in [1.2, 7.8], i \in [4]$. 

The environment contains 30 obstacles. We consider 174 different final positions. Table \ref{tab:3} lists the resulting solve times of the nonlinear program. The approximate method solved 8.7x faster than the exact method with respect to median solve times. In 9 instances the exact method failed or exceeded the maximum solve time of 20s. The approximate method was less than 5\% sub-optimal for 155 of the remaining 165 test cases. The worst-case was 16\% sub-optimal. The upper plot of Figure \ref{fig:quadcopter} shows the quadcopter navigating the cluttered environment. The lower plot gives the configuration-space view with the Minkowksi sum approximations shown in yellow.

\begin{table}[]
\vspace{1.5mm}
\begin{center}
\caption{Solve Times Statistics for Quadcopter Example}
\begin{tabular}{||c | c c c||}  
 \hline
  Collision Avoidance & Min. (s) & Median (s) & Max. (s) \\ [0.5ex] 
  \hline
 Approximate (4th-Order) & 0.47 & 0.74 & 2.76 \\ 
 \hline
 Exact & 2.37 & 6.48 & 18.89 \\
 \hline
\end{tabular}
\label{tab:3}
\end{center}
\vspace{-1.5mm}
\end{table}

\begin{figure}[]
    \centering
    \vspace{0.2cm}
    \includegraphics[width=0.485\textwidth]{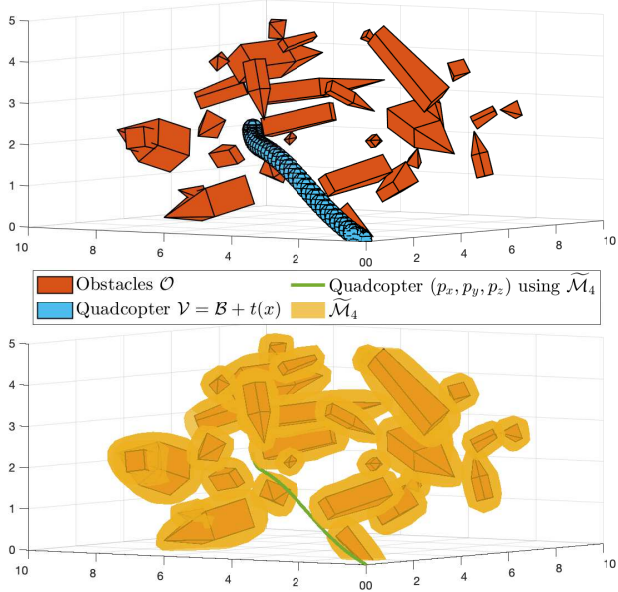}
    \caption{Quadcopter navigating in workspace (upper) and C-space (lower)}
    \label{fig:quadcopter}
\end{figure}

\begin{figure*}[]
    \centering
    \vspace{0.2cm}
    \includegraphics[width=0.99\textwidth]{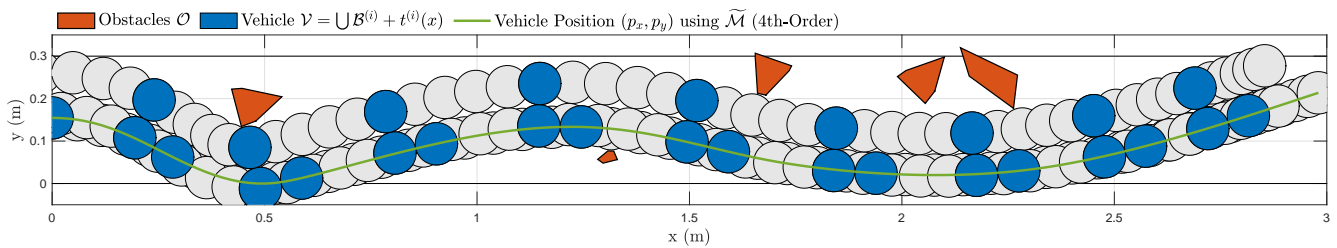}
    \caption{L-shaped autonomous car navigating obstacles in workspace}
    \label{fig:car_traj_multiball}
\end{figure*}
\subsection{Autonomous Car with Multiple-Disc Geometry}
The previous examples used vehicle geometries consisting of a single Euclidean ball. We briefly demonstrate how our method can handle vehicle geometries consisting of multiple Euclidean balls. Returning to the autonomous car example, we introduce two additional discs, each with radius 0.05m, to form an ``L'' shape. The center of each disc is a function of the vehicle's position $(p_x,p_y)$ and orientation $\psi$:
\begin{equation}
    t^{(i)}(x) = \begin{bmatrix} p_x \\ p_y \end{bmatrix} + 
    \begin{bmatrix}
    \cos{\psi} & -\sin{\psi} \\
    \sin{\psi} & \cos{\psi}
    \end{bmatrix}
    \begin{bmatrix}
    l^{(i)}_x \\ l^{(i)}_y
    \end{bmatrix}
\end{equation}
where 
\begin{equation}
   l^{(1)}_{x,y} = \begin{bmatrix}
   0 \\ 0
   \end{bmatrix},  \;
      l^{(2)}_{x,y} = \begin{bmatrix}
   0.05 \\ 0
   \end{bmatrix}, \;
      l^{(3)}_{x,y} = \begin{bmatrix}
   0 \\ 0.05
   \end{bmatrix},
\end{equation}
Because the radius of each disc is identical, we can reuse the same approximation $\widetilde{\mathcal{M}} = \{ y \in \mathbb{R}^2 \, |\, p^{(m)}(y) \leq 1 \}$ for the Minkowski sum of a given obstacle $m$ and a disc of radius 0.05m.\footnote{If this were not the case we would have to compute separate approximations for each unique radius value.} We simply change the argument $t^{(i)}(x)$ supplied to $p$, representing the center of the ball $i$ as a function of the vehicle's state $x$. In this setting \eqref{eqn:polyobs} is replaced with the following:
\begin{equation}
    p^{(m)}(t^{(i)}(x_k)) > 1, \, k \in[N], \, m \in[M], \, i \in [3]. \label{eqn:polyobs_multiball}
\end{equation}
Figure \ref{fig:car_traj_multiball} shows the L-shaped vehicle navigating obstacles.

\subsection{Scaling}
One caveat of using SOS polynomials to approximate indicator functions of sets is that the polynomial may return large values for points far outside of the set. This may be problematic for NLP solvers which are often sensitive to the scaling of the problem. To improve the scaling, we can apply any smooth function $q(z): \mathbb{R}_{\geq 0} \rightarrow \mathbb{R}$ to \eqref{eqn:polyobs} which is strictly-increasing for $z \geq 0$. In the examples shown we have utilized $q(z) = -\textup{exp}(-z)$, replacing \eqref{eqn:polyobs} with:
\begin{equation}
    -\textup{exp}(-p^{(m)}(t(x_k))) > -\textup{exp}(-1), \quad k \in[N], \quad m \in[M].
\end{equation}
As $p$ is a SOS polynomial, the left-hand side of this equivalent formulation takes on values in the range $[-1, 0]$. 
\subsection{Implementation Details}
All examples were solved on a MacBook Pro with a 2.6 GHz 6-Core Intel Core i7 CPU. YALMIP \cite{Lofberg2004} was used in conjunction with MOSEK \cite{Mosek2017} to solve the SOS optimization problems. IPOPT \cite{Wachter2006} with the MA27 linear solver was used to solve the nonlinear optimization problems with exact gradients and Hessians supplied by CasADi \cite{Andersson2019}. Supporting code is available at https://github.com/guthriejd1/mink-obca.

\section{CONCLUSIONS}
This work presented novel obstacle avoidance conditions based on outer approximations of Minkowski sums. This method is advantageous in that it yields a much smaller nonlinear program compared to exactly representing the collision avoidance conditions. On motion planning problems for an autonomous vehicle and quadcopter, the approximate method solved 4.8x and 8.7x faster respectively when navigating cluttered environments. The resulting trajectories exhibited minimal sub-optimality compared to using exact collision avoidance conditions.
Currently our method is limited to cases in which the vehicle is represented by a union of Euclidean balls and the obstacle is a bounded, convex polytope. In future work we plan to consider representations of non-convex obstacles.

\addtolength{\textheight}{-12cm}   

\bibliography{myref}
\bibliographystyle{ieeetr}
\end{document}